\documentclass[prx,prl,a4paper,aps,twocolumn,superscriptaddress,longbibliography]{revtex4-2}
\usepackage[T1]{fontenc}
\usepackage[utf8]{inputenc}
\usepackage{amsmath, braket, empheq}
\usepackage{amsthm}
\usepackage{amssymb}
\usepackage{graphicx}
\usepackage{amsfonts}
\usepackage[english]{babel}
\usepackage{fancyhdr} 
\usepackage{array} 
\usepackage{color} 
\usepackage{CJK}
\usepackage{mathtools}
\usepackage{subfigure}
\usepackage[table]{xcolor}
\usepackage{collcell}
\usepackage{geometry}
\usepackage{listings} 

\usepackage{placeins, bbm}

\theoremstyle{plain}
\newtheorem{thm}{\protect\theoremname}

\newtheorem{lem}{Lemma}

\providecommand{\theoremname}{Theorem}
\newcommand*{\myproofname}{Proof}

\newcommand*{\id}{\mathbbm{1}}
\newcommand*{\Tr}{\textrm{Tr}}
\newcommand*{\fin}{\textrm{fin}}
\newcommand*{\In}{\textrm{in}}

\usepackage{tikz}
\definecolor{myurlcolor}{rgb}{0,0,0.7}
\definecolor{myurlcolor1}{rgb}{0,0.7,0.1}
\definecolor{myrefcolor}{rgb}{0,0,0.7}
\usepackage[unicode=true,pdfusetitle, bookmarks=true,bookmarksnumbered=false,bookmarksopen=false, breaklinks=false,pdfborder={0 0 0},backref=false,colorlinks=true, linkcolor=myrefcolor,citecolor=myurlcolor1,urlcolor=myurlcolor]{hyperref}

\makeatletter

\begin{document}

\title{Thermodynamics of quantum switch information capacity activation}
     \author{Xiangjing Liu }
     \email{liuxj@mail.bnu.edu.cn}
\affiliation{Shenzhen Institute for Quantum Science and Engineering and Department of Physics, Southern University of Science and Technology, Shenzhen 518055, China}
\author{Daniel Ebler}\email{Ebler.Daniel1@huawei.com}
\affiliation{Theory Lab, Central Research Institute, 2012 Labs, Huawei Technology Co. Ltd., Hong Kong Science Park, Hong Kong SAR}
\affiliation{Department of Computer Science, The University of Hong Kong, Pokfulam Road, Hong Kong SAR}
\author{Oscar Dahlsten}\email{dahlsten@sustech.edu.cn}
\affiliation{Shenzhen Institute for Quantum Science and Engineering and Department of Physics, Southern University of Science and Technology, Shenzhen 518055, China}

\begin{abstract}  We address a new setting where the second law is under question:  thermalizations in a quantum superposition of causal orders, enacted by the so-called quantum switch. This superposition has been shown to be associated with an increase in the communication capacity of the channels, yielding an apparent violation of the data-processing inequality and a possibility to separate hot from cold. We analyze the thermodynamics of this information capacity increasing process. We show how the information capacity increase is compatible with thermodynamics. We show that there may indeed be an information capacity increase for consecutive thermalizations obeying the first and second laws of thermodynamics if these are placed in an indefinite order and moreover that only a significantly bounded increase is possible. The increase comes at the cost of consuming a thermodynamic resource, the free energy of coherence associated with the switch. 
\end{abstract}

\maketitle

\noindent {\bf {\em General Introduction---}}The study of the Laws of thermodynamics, including the First Law of energy conservation and the Second Law of free energy non-increase, are key to the foundations of physics~\cite{lieb2013entropy,  janzing2000thermodynamic}. Thermalizing evolutions (thermal channels) are central to thermodynamics, representing the impact of a heat bath on the system of interest~\cite{ horodecki2013fundamental, janzing2000thermodynamic, brandao2013resource, lostaglio2015quantum, brandao2015second}. Thermalizing evolutions will in general remove information about the system's earlier interactions, in line with the second law of thermodynamics. One may think of them as noisy communication channels from the past to the future. How much information can be conveyed is captured by the information capacity of the channel, a central concept in the Shannon information theory \cite{wilde2013quantum}.

Recent results \cite{ebler2018enhanced,salek2018quantum,procopio2019communication,caleffi2020,chiribella2021indefinite,chiribella2021quantum} suggest the information capacity of two consecutive channels may be increased if the order of the channels is in a quantum superposition of the two possible orders, an operation termed applying a quantum switch to the two channels~\cite{Chiribella2013}. The quantum switch extends the allowed operations in quantum communication theory~{\cite{chiribella2019quantum,kristjansson2019}, and is an example of how to create indefinite causal order \cite{oreshkov2012quantum}, an intriguing area of possible new physics~\cite{hardy2009quantum,belenchia2018quantum,henderson2020quantum}.

Because of the strong connection between information theory and thermodynamics, e.g.\ via the central role played by entropy, this phenomenon may have significance for thermodynamics. The setting wherein information capacity is increased via the quantum switch was shown to give an apparent violation of the data-processing inequality~\cite{salek2018quantum}, which is closely connected to the second law~\cite{cover1999elements}, as well as enable the separation between hot and cold~\cite{felce2020quantum, nie2020experimental, cao2021experimental}. A more general study of superpositions of thermodynamic processes is beginning~\cite{rubino2021quantum,rubino2021inferring, wood2021operational,guha2020thermodynamic}.
These intriguing results raise the question of whether the laws of thermodynamics and the information capacity increase of the switch are compatible, and if so how. Is the increased communication capacity allowed within thermodynamics, if so is there a limit to how much is allowed? Another concrete open question, posed in~\cite{felce2020quantum} concerns whether the switch is a thermodynamic resource or should be regarded as thermodynamically free. 

In this work we accordingly investigate the interplay between the laws of thermodynamics and this information capacity increase. We find that switched thermal  channels may have a small increase in information capacity, and no more, or the laws of thermodynamics would be violated. To show this we derive an {\em upper}-bound on the mutual information (a measure of communication amount) that can be established in this setting, depicted in Figure~\ref{fig:circuitmodels}. The bound  captures how activation is only possible when the switch is ON. We argue that this increase in capacity comes at a thermodynamical cost. The switch can, when a careful accounting is done, create free energy. Rather than interpreting our results as showing that the switch violates the second law of thermodynamics, we argue it should not be viewed as being in the domain of applicability of the second law: the quantum switch is a thermodynamical resource. 
  \begin{figure} 
   \centering
  \includegraphics[width=0.4\textwidth]{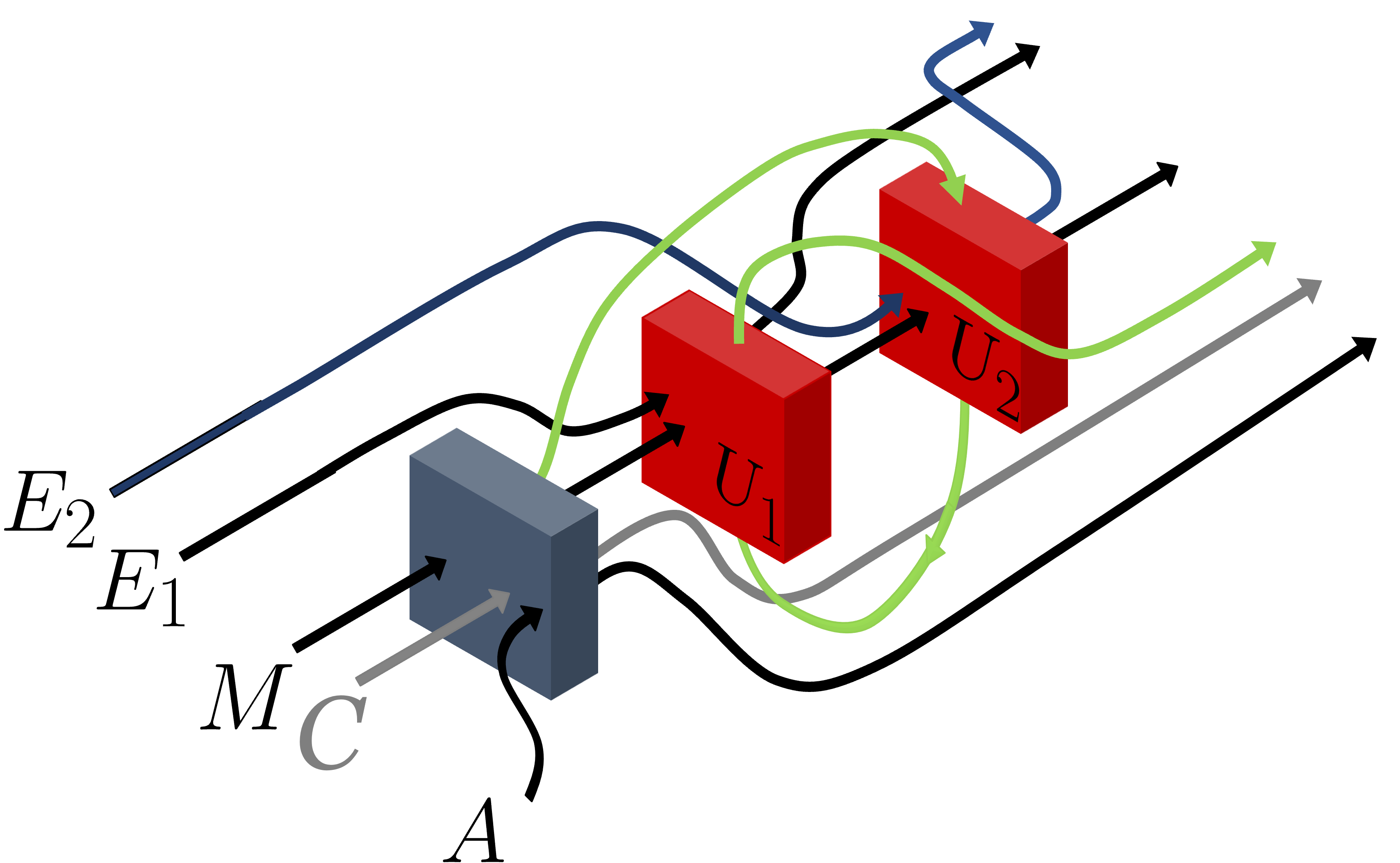} 
    \caption{The set-up. $M$ carries information about the system $A$, $E$ is from the heat bath environment. There is a unitary interaction between $M$ and $E$. If the switch is ON there is a superposition of the black and green curved-path orders for how the Kraus operators of the channel are applied to $M$, else there is a fixed order as per the black path.} 
    \label{fig:circuitmodels} 
\end{figure}

\noindent {\bf {\em Laws of thermodynamics and thermal operations---}}We begin with describing thermodynamical concepts and quantities that are key to our results. A central concept is the (Gibb's) thermal state density matrix $\tau=\exp(-H/kT)/Z$ where $H$ is the system's Hamiltonian, $T$ the temperature, $k$ Boltzmann's constant and $Z=\Tr (\exp(-H/kT))$ a normalization constant. A {\em heat bath} is normally assumed to be in the state $\tau$ as well as being so large that any system interacting with it will eventually reach a thermal state at the temperature of the bath. 

The Second Law for a system interacting with a heat bath can be presented as demanding that the relative entropy $S(\rho || \sigma ):=\Tr\left(\rho(\log\rho-\log\sigma)\right)$ between the state $\rho$ and the thermal state $\tau$ of the given system cannot increase under the heat bath's influence $\varepsilon(\rho)$, i.e., 
 \begin{align}
\label{eq:2ndlaw}
 S(\varepsilon(\rho) || \tau )- S(\rho || \tau ) \leq 0.
 \end{align} 
The free energy is a crucial related quantity: $F(\rho,H,T):= U(\rho)-kTS(\rho)=\Tr(\rho H)-kTS(\rho) $, for a system with state $\rho$, Hamiltonian $H$, temperature $T$ and the von Neumann entropy $S(\rho) :=- \Tr( \rho \ln \rho)$. The relative entropy (with natural logarithm) to the thermal state can be equivalently written in terms of the free energy: 
\begin{equation}
\label{eq:freeenergy}
kTS(\rho || \tau)=F(\rho) -F(\tau). 
\end{equation}
Thus the second law of Eq.\eqref{eq:2ndlaw} can be restated as the free energy gap to the thermal state not increasing: $\Delta F\leq 0$. Moreover, the thermal state $\tau$ has minimal free energy, so the free energy cannot increase (whereas entropy alone can decrease e.g.\ in cooling scenarios). 

The long-established terminology of calling Eq.~\eqref{eq:2ndlaw} or similar conditions the second {\em Law} can be considered to be confusing. There are maps, physically and mathematically,  which violate Eq.~\eqref{eq:2ndlaw} such as $\varepsilon(\rho)=\ket{1}\bra{1}$ with $\ket{1}$ not being the ground state (and thus not a thermal state for any temperature). A safe minimal interpretation of 
 Eq.~\eqref{eq:2ndlaw} is that it is a property that only some dynamical maps have. We demand that interactions with systems deemed to be heat baths should have this property. Then, for logical consistency, a map violating Eq.~\eqref{eq:2ndlaw} should not correspond physically to an interaction with a heat bath. 
  
We will moreover refer to the idea of {\em free} operations and {\em resources} (which are not free)~\cite{lieb2013entropy}. A key paradigm in thermodynamics is that systems at some ambient temperature $T$ are free (e.g.\ not requiring the burning of fuel) and anything that is not free is a resource. $kTS(\rho || \tau)=F(\rho) -F(\tau)$ may be thought of as the {\em currency} of thermodynamics, quantifying the resource value of a given system's state $\rho$. Accordingly, it is natural to require Eq.~\eqref{eq:2ndlaw} to be respected in order for a map to be termed thermodynamically free.

Quantum channels (completely positive and trace-preserving maps) $\varepsilon(\cdot)$ which leave the thermal (Gibbs) state invariant, $\varepsilon(\tau)=\tau$, are called {\em thermal} or {\em Gibbs-preserving}. The Gibbs-preserving condition is equivalent to Eq.~\eqref{eq:2ndlaw} (see Appendix~\ref{app:thermal2ndlaw}).  Quantum channels, including Gibbs-preserving channels, are moreover mathematically equivalent to unitarily interacting the system with an environment $E$ and then tracing over $E$~\cite{nielsen2002quantum}. To model a Gibbs-preserving channel this way, one may take $E$ to be in a thermal state initially, and apply an energy-preserving unitary $U_{ME}$ (meaning $[U_{ME}, H_M+H_E]=0$) followed by tracing over $E$. This is commonly called a thermal operation~\cite{horodecki2013fundamental, brandao2015second, brandao2013resource, janzing2000thermodynamic,cwiklinski2015limitations, lostaglio2015quantum}. (We shall later consider thermal operations in superposed orders).

\noindent{\bf {\em Quantum Switch---}}The quantum switch $\mathcal{S}$ is a supermap~\cite{chiribella2008quantum}, which takes two quantum channels $ \mathcal{C}_1,\mathcal{C}_2 $ as input and outputs another quantum channel $ \mathcal{S} ( \mathcal{C}_1, \mathcal{C}_2 )$.  Denoting the Kraus operators~\cite{nielsen2002quantum} of the channel $ \mathcal{C}_1 $ as $\{ K^{(1)}_i\}$ and $ \mathcal{C}_2 $ as $ \{ K^{(2)}_j \} $, the switched channels  $ \mathcal{S}_{\sigma_C} ( \mathcal{C}_1, \mathcal{C}_2 )$ act as
\begin{align}\label{eq: actionofswitch}
 \mathcal{S}_{\sigma_C} ( \mathcal{C}_1, \mathcal{C}_2 ) (  \rho  )= \sum_{ij} S_{ij} (\sigma_C \otimes \rho  ) S_{ij}^\dag,
\end{align}
where $\sigma_C$ is the state of the control system (also called the trajectory in the literature \cite{chiribella2019quantum,kristjansson2019}) and $S_{ij}$ denotes the Kraus operators for $ \mathcal{S}_{\sigma_C} ( \mathcal{C}_1, \mathcal{C}_2 )$, namely,
\begin{align}
S_{ij}=|0\rangle_C \langle 0| \otimes K^{(2)}_j K^{(1)}_i +  |1\rangle_C \langle 1| \otimes K^{(1)}_i K^{(2)}_j \ .
\end{align}
The quantum switch superposes the two temporal orders of the channels -- that is $\mathcal{C}_2 \circ \mathcal{C}_1$ and $\mathcal{C}_1 \circ \mathcal{C}_2$ -- with amplitudes determined by $\sigma_C$~\cite{Chiribella2013}.

The state $\sigma_C$ is very important. If $\sigma_C=\ket{+}\bra{+}$ there is a superposition of causal orders but e.g.\ if $\sigma_C=\ket{0}\bra{0}$, there is a well-defined order. To interpolate between these two extremes, we shall allow for 
$\sigma_C = \lambda \ket{+}\bra{+} + (1-\lambda ) \ket{0}\bra{0}, $
in which the free parameter $\lambda \in [0,1] $ describes to what extent the switch is ON. The alternative choice of $\ket{\psi}_C=\sqrt{\frac{\lambda}{2}}\ket{0}+\sqrt{1-\frac{\lambda}{2}}\ket{1}$ yields qualitatively similar results (Appendix A9).

\noindent {\bf {\em Model---}}We now describe the setting, depicted in Figure~\ref{fig:circuitmodels}. There is a two-level system $M$ which carries information about a system $A$. $M$ will be thermalized via unitary interactions with the $E_i$ qubits from the heat bath. We take the Hamiltonians to be $H_{E_j}=-\sigma_z\,\forall j$, and $H_M=-\vec{n} \cdot \vec{\sigma}$, where $\vec{n}$ is a three-dimensional unit vector and $\sigma$ indicates Pauli-matrices. The environmental qubits $E_i$ are initially in their thermal states  $\tau_{E_i}=q \ket{0}_{E_i}\bra{0}+ (1-q) \ket{1}_{E_i}\bra{1} $  with $q:= e^{\frac{1}{kT}}/Z \in [ 1/2, 1]$. The compound system $AM$ is initially uncorrelated with $E$ and is, for simplicity, in the state:
\begin{equation}
  \rho^\In_{AM}= p \rho_A^{(0)}\otimes \ket{\phi }_{M}\bra{\phi}+(1-p)  \rho_A^{(1)}\otimes \ket{\phi^\perp}_{M}\bra{\phi^\perp} ,
\end{equation}
 where $\ket{\phi}, \ket{\phi^\perp}$ denote the eigenstates of $H_M$.
 
To determine the dynamics, we prove (see Appendix~\ref{proof: energyconsv}), inspired by Ref.~\cite{scarani2002thermalizing},  that the energy conservation condition 
\begin{align} \label{eq: energyconserv}
  [U_{ME_j}, H_M+H_{E_j}]=0,
  \end{align}
gives a specific class of unitary evolutions for qubit collisions:
 \begin{align}\label{eq:collisionU}
 U_{ME_j}(\theta_1, \theta_2 ) = ( \tilde{U}\otimes \mathbb{I}) e^{i \theta_1 S }e^{i\theta_2\frac{ \sigma_z \otimes \sigma_z}{2}}   (\tilde{U}^\dag\otimes \mathbb{I}),
 \end{align}
 where $\tilde{U} $ is a  one-qubit unitary such that $\tilde{U} \sigma_z  \tilde{U}^\dag =\vec{n} \cdot \vec{\sigma}$ and 
the swap operator $S (\ket{a}\ket{b})= \ket{b}\ket{a}$. The associated Kraus operators on $M$ do not commute (see Appendix~\ref{krausdecomp}). Since $U(\theta):= e^{i \theta S}=\cos(\theta) \mathbb{I} +i \sin(\theta) S$, we see $\sin(\theta):=s\in [0,1]$ quantifies the thermalisation strength (and denote $\cos(\theta):=c)$.

The overall interaction, as depicted in Figure \ref{fig:circuitmodels}, is then the switched version of the unitaries $U_{ME_1}(\theta)$ and $U_{ME_2}(\theta)$: 
 \begin{align}\label{totalunitary}
 L\!=\!\ket{\psi }_c\bra{\psi } \otimes U_{ME_2} U_{ME_1} \!+\! \ket{\psi^\perp }_c\bra{\psi^\perp}\otimes U_{ME_1} U_{ME_2} ,
 \end{align}
 where $ \ket{\psi}, \ket{\psi^\perp} $ are eigenvectors of $H_C$. $L$ is unitary and energy-preserving (see Appendix~\ref{sec: energyconsv}).

\noindent {\bf {\em Bound on information capacity from thermodynamics---}} We now derive an upper bound 
 on the mutual information of channels that are Gibb's-preserving and energy-preserving. For classical channels the mutual information $I(A:B)$ between an input-record A and output B (optimized over inputs $A$) can be interpreted as the optimal communication rate, termed the channel's information capacity~\cite{cover1999elements}. The quantum mutual information of the input-record and the output of a quantum channel, $I(A:B):=S(A)+S(B)-S(AB)$, can be interpreted similarly~\cite{wilde2013quantum}. Quantum channels followed by measurements induce classical channels whose optimized mutual information (the so-called accessible information of the quantum channel~\cite{wilde2013quantum}) is upper bounded by the optimized pre-measurement $I(A:B)$. Thus our bound on $I(A:B)$ can be interpreted as a bound on the classical information capacity.

\begin{thm}[Bound on information capacity from thermodynamics]
\label{thm:main} Consider the mutual information between the message system $M$ and a record of the message $A$. Initially it is $I(A: M^\In)$ and finally, after the quantum switched energy-preserving and Gibbs-preserving channels, it is $I(A: CM^{ \fin })$, where $C$ is the control system. Then 
 \begin{align}
 I(A: CM^{ \fin }) \leq & ( c^4+ \lambda  \left(q^2-q+\frac{1}{2}\right)s^4)  I(A: M^\In) \nonumber\\
 &  +\lambda G_{\geq 0},
\end{align}
where $\lambda\in [0,1]$ represents how much the switch is ON, $s$ is the thermalisation strength ($s^2+c^2=1$) and $q$ is the ground state probability of the environment thermal state. 
$G_{\geq 0}$ is a non-negative function of $s$ and $q$ given in Eq.~\eqref{eq: virtualmutual} of Appendix~\ref{infocapbound}.
 \end{thm} 
Before describing the proof, we note that the bound captures that a certain amount of increase in mutual information is possible from from the switch being ON: the quantities multiplied by $\lambda$. Moreover only a limited amount of increase is allowed from the switch. Two extreme temperature cases illustrate this: (i) if the temperature of the heat bath is infinite ($q=1/2$),  
\begin{equation}
I(A:CM^\text{fin}) \leq (c^4+  \frac{ \lambda }{4} s^4) I(A:M^\text{in }),
\end{equation}
 and (ii) if the temperature of the heat bath is zero ($q=1$), 
 \begin{equation}
  I(A: CM^{ \text{fin} })  \leq   ( c^4+  \frac{\lambda}{2}s^4)  I(A: M^\text{in}) +\lambda G_{\geq 0}.
 \end{equation} 
It follows, for example, that for for infinite temperature, energy conserving Gibbs-preserving channels can at most have an increase $\frac{ \lambda }{4} s^4 I(A:M^\text{in })\leq \frac{1 }{4}I(A:M^\text{in })$. We moreover see that the low temperature allows for more information capacity as well as a greater difference between the cases where the switch is ON and OFF. This is consistent with how for the infinite temperature case only, $I(A:C)=0$ (as per Eq.~\eqref{eq: mutualAC} in the Appendix).  The bound is tight for $s=0$ (see Figure~\ref{fig:simulations}). An example of a pair of channels--one of which is not energy-preserving--{\em violating} the bound is given in Appendix~\ref{app:decoherent}, demonstrating that the bound adds a significant restriction.
  \begin{figure}    
   \centering
  \includegraphics[width=0.48\textwidth]{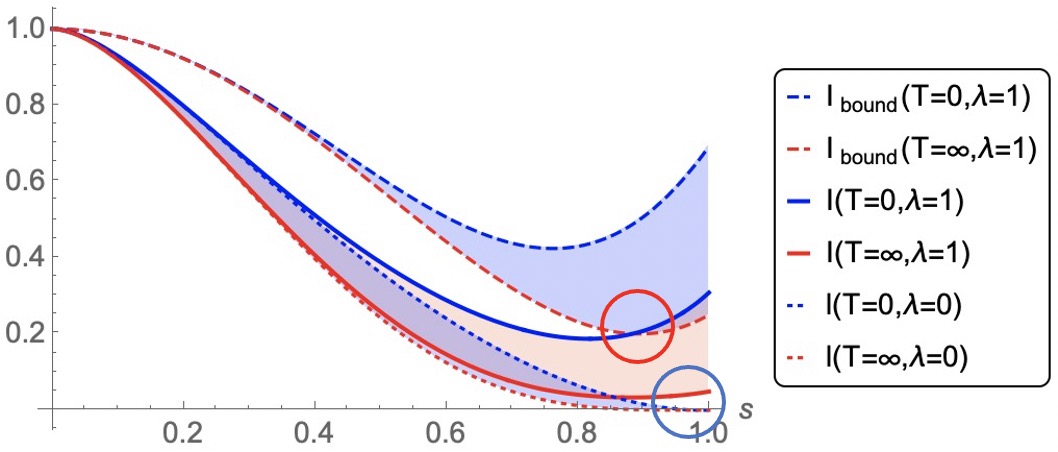} 
    \caption{We derived, from the first and second laws, a general bound on the final mutual information between a copy of the input message and the output after the switch $I(A:CM^{\fin})$. Here we show the bound together with a curve that is realizable for two different temperatures. $s\in [0,1]$ is the thermalization strength and $\lambda \in [0,1]$ the extent to which the switch is ON. The red circle proves that low temperature can outperform high temperature in terms of mutual information and the blue circle illustrates the increased mutual info when the switch is ON.} 
    \label{fig:simulations} 
\end{figure}

To illustrate the approach to proving Theorem~\ref{thm:main} (see Appendix~\ref{infocapbound} for the full proof), consider the case given in~\cite{ebler2018enhanced} corresponding to an infinite temperature heat bath with the quantum switch fully ON ($q=1/2,\lambda=1$). Firstly apply the dynamics given in Eq.~\eqref{totalunitary} to the initial state $\sigma_C \otimes \rho^\In_{AM}$. The final state $\rho^\fin_{CAM}$ then turns out to obey
\begin{align}
 \rho^\fin_{CAM} =& \ket{-}_C\bra{-} \otimes \left( k \left(\sigma_x^M \rho^\In_{AM} \sigma_x^M\right) +   l \left(\rho^\In_A  \otimes \frac{ \mathbb{I}_M }{2} \right) \right)\nonumber \\  
 +& \ket{+}_C\bra{+} \otimes \left(m (\rho^{\text{in }}_{AM})+n   (\rho^\In_A\otimes \tau_M)  \right),
\end{align} 
 where $\rho_A^\In:=\Tr_M \rho^\In_{AM}$ and $k$,$l$,$m$ and $n$ are all non-negative real numbers. This makes the state amenable to conditioning on $C$, such that one can employ the mutual information (chain) rule
\begin{align}\label{eq: mchainrule}
I(A:CM^\fin)=I(A:C) + I(A:M^\fin |C).
\end{align}
This breaks the task into calculating $I(A:C)$ as well as $I(A:M^\fin |C)=p(\ket{+}_C) I(A:M^\fin |\ket{+}_C)+p(\ket{-}_C)I(A:M^\fin |\ket{-}_C)$.
The reduced states in question turn out to be mixtures of product states with another state with the same reduced state on A, i.e.
$\rho^{\text{fin}}_{AB}=p\rho^{\text{in}}_{AB}+ (1-p)\rho^{\text{in}}_A \otimes \sigma_B $. We show (via the data processing inequality) that for such states
\begin{align}
\label{eq:prodstatebound}
I(A:B^\text{fin}) \leq p I(A:B^\text{in}) .
\end{align}
We use Eq.~\eqref{eq:prodstatebound} to upper bound the two terms in the RHS of Eq.~\eqref{eq: mchainrule}. (The argument details, including the extension to arbitrary temperature of the heat bath and $\lambda \in [0,1]$, are given in Appendix~\ref{infocapbound}).

\noindent {\bf {\em Information capacity increase consumes a thermodynamical resource---}}We now argue that $\mathcal{S}_{\sigma_C} ( \mathcal{C}_1, \mathcal{C}_2 ) (  \cdot  )$, where  $\mathcal{C}_1$ and $\mathcal{C}_2$ are both Gibbs-preserving channels leaving the same thermal state invariant, should not be regarded as thermodynamically free,  except for the case of $\sigma_C=\tau_C$, the thermal state of $C$ for the ambient temperature. To formalise this argument we firstly consider the impact of $\mathcal{S}_{\sigma_C}$ on $M$, before including the impact on $C$. 
 
\begin{thm}[Switched Gibbs-preserving channels is a Gibbs-preserving channel on $M$] \label{superchan}
If $\mathcal{C}_1(\tau_M)=\mathcal{C}_2(\tau_M)=\tau_M$, where $\tau_M$ is the thermal state on $M$ then, for any state $\sigma_C$, $Tr_C\mathcal{S}_{\sigma_C} ( \mathcal{C}_1, \mathcal{C}_2 ) ( \tau_M )=\tau_M$. 
\end{thm}
The proof of theorem~\ref{superchan}  can be found in Appendix~\ref{proofsuperchan}. A similar claim was made independently in Ref.~\cite{guha2020thermodynamic}.

Theorem \ref{superchan} leaves it open for  $\mathcal{S}_{\sigma_C}$ to be viewed as a thermodynamically free operation. However, it is important to also take the presence of $C$ into account. 

\begin{thm}[Including the control system] \label{thermalcontrol}
 (i) For the control system initially in a thermal state $\tau_C$, and $\mathcal{C}_1, \mathcal{C}_2$ being Gibbs-preserving channels, $\mathcal{S}_{\tau_C}( \mathcal{C}_1, \mathcal{C}_2 )$ leaves $\tau_C \otimes \tau_M$ invariant. (ii) For the control system initially in $\ket{+}\bra{+}$, and $\mathcal{C}_1, \mathcal{C}_2$ being Gibbs-preserving channels for infinite temperature, $\mathcal{S}_{\ket{+}}( \mathcal{C}_1, \mathcal{C}_2 )$ followed by a particular free operation does not leave the thermal state on $M$, $\tau_M$ invariant. 
\end{thm} 

Statement (i) in Theorem \ref{thermalcontrol} follows because $\tau_C \otimes \tau_M$  can be interpreted as a mixture of two cases $\sigma_C=\ket{0/1}\bra{0/1}$. By inspection, from the definition of the switch $\mathcal{S}$, in each case $\sigma_C$ is unchanged by the switch. Moreover the thermal state on $M$ is unchanged since only Gibbs-preserving channels are applied to $M$. Hence it appears fair to call $\mathcal{S}_{\tau_C} ( \mathcal{C}_1, \mathcal{C}_2 ) (\cdot)$ a  thermodynamically free operation, at least in so far as Eq.\eqref{eq:2ndlaw} is demanded. 

If the control system is instead in $\ket{+}$ as in statement (ii), things are subtly but importantly different. Consider for simplicity the case of infinite temperature (or null Hamiltonians) and full swap ($s=1$). Whilst Theorem \ref{superchan} implies $\tau_M$ is invariant it is a priori not sensible to ask whether $\tau_C\otimes \tau_M$ is invariant in this case, since $\sigma_C\neq \tau_C$. However we can answer the question of whether it is a free operation in another way: acting on the final $CM$ with a certain free operation changes $M$ into something that is different to $\tau_M$. The swap unitary $S_{CM}$ should be treated as free here, as discussed in the model section. Now, 
\begin{align}
&\Tr_C(S_{CM}\circ \mathcal{S}_{\ket{+}} (\mathcal{C}_1, \mathcal{C}_2 ) (\tau_M))\nonumber \\
&=\frac{3}{8}\ket{+}_M\bra{+}+\frac{5}{8}\ket{-}_M\bra{-}
\neq \frac{1}{2}\id = \tau_M,
\end{align} 
as shown in Eq.~\eqref{eq: lambda1}. Thus the switched Gibbs-preserving channels under such an initial state on the control system should not be termed a free thermal operation (else one would need to declare a violation of the Second Law of Eq.~\eqref{eq:2ndlaw}).

The above argument is analogous to Bennett's famous exorcism of the Szilard-engine Maxwell's demon~\cite{bennett2003notes}. Bennett's argument contains two (qu)bits (sometimes one of them is a trit but for simplicity we take two qubits): the demon's memory $A$ and the working medium system $B$. Initially $\rho_{AB}=\ket{0}_A\bra{0}\otimes \id_B/2$. The demon performs its famous measurement modelled as a $\mathrm{CNOT}=\id_A\otimes \ket{0}_B\bra{0}+(\sigma_x)_A\otimes \ket{1}_B\bra{1}$. Bennett notes, crucially\cite{bennett2003notes}, that the CNOT operation may in principle be thermodynamically free, considering that the initial and final Hamiltonian may be null ($H=0$), and noting that CNOT leaves the corresponding thermal state ($\frac{1}{4}\id_{AB}$) invariant. There can now, analogously to the swap unitary above, be a second free operation, namely CNOT in the opposite direction, setting $B$ to $\ket{0}$ regardless of its initial state. Thus, the end effect is a map on $B$ respecting $\varepsilon(\rho_B=\frac{1}{2}\id=\tau_B)=\ket{0}_B\bra{0}\neq \tau_B$.  Such a map creates free energy, e.g.\ taking a thermal state to a non-thermal state, and should be treated as a resource. 

The switch can be said to harness only one type of free energy from $\sigma_C$,  known as the free energy of coherence. If $\sigma_C$ is diagonal in the computational basis, then, as discussed above, the action of the switch is trivial. Instead there must be coherence. The free energy of coherence can be defined as  $F_{coh}=F(\sigma)- F(\mathcal{D}_H(\sigma))$, where $\mathcal{D}_H(\cdot)$ diagonalises the state in the Hamiltonian eigenbasis~\cite{lostaglio2015description}. From the definition, $F=F_{coh}+F(\mathcal{D}_H(\sigma))$. $F_{coh}(\sigma_C)$ is a natural quantifier of the thermodynamic resource value of the switch. We can explicitly mathematically capture this by noting that $
F_{coh}(\sigma, H) \geq kT \frac{\lambda^2}{\ln16}$ (see Appendix A8) such that the information capacity activation in Theorem~\ref{thm:main} is directly governed by  $F_{coh}$.

\noindent{\bf {\em Relation to existing work--}}The conclusion that the switch is a thermodynamical resource does not contradict the resource theory created in Ref.~\cite{kristjansson2019}, in which the switch is taken as free, but rather shows that these are separate resource theories. The result that the switch is not thermodynamically free is consistent with the conclusion from Ref.~\cite{felce2020quantum} that it could be possible to use the switch to perform refrigeration.  The results are also consistent with Ref.~\cite{salek2018quantum} which show how the switch gives an apparent violation of the bottleneck inequality,  which holds for all channels. The bottleneck inequality (which is equivalent to the data processing inequality) for two maps $\varepsilon_1(\cdot)$ and $\varepsilon_2(\cdot)$ says that $I(A:\varepsilon_2\circ\varepsilon_1(B))\leq  I(A:\varepsilon_1(B))$. The information capacity activation phenomenon can be expressed as $I(A:\mathcal{S}_{\ket{+}}(\mathcal{C}_1, \mathcal{C}_2)(B))>  I(A:\mathcal{C}_1(B))$, indicating a possible violation of the bottle-neck inequality. As a vivid example, let $\mathcal{C}_2=\mathcal{C}_1:=\mathcal{C}$ correspond to the full swap case (s=1) such that $\mathcal{C}(\tau_M)=\tau_M$. If we denoted, carelessly, the action on $M$ of the switch channels followed by the conditioning on $C$, as $\mathcal{C}\circ \mathcal{C}$, there would indeed be a violation of Eq.~\eqref{eq:2ndlaw} in the sense that in the second application of $\mathcal{C}$ the state goes from being thermal to becoming athermal. This again highlights how the switch should be treated as a thermal resource.

\noindent{\bf {\em Summary and Outlook--}} We derived an upper bound on the information capacity of energy-preserving Gibbs-preserving channels with or without indefinite causal order. The bound captures the capacity activation associated with indefinite causal order as well as revealing that this is stronger for low temperatures. We argued that there is a thermodynamic cost to this activation (or else the Second Law is violated). 

It remains to be investigated whether the arguments here directly generalise to single-shot thermodynamics where statements concern other quantities than average work, e.g. the worst-case work~\cite{DahlstenRRV11, del2011thermodynamic, horodecki2013fundamental, aaberg2013truly}. It is intriguing that only a limited activation capacity from the switch is allowed for Gibbs-preserving channels, deserving further investigation. Our approach may help clarify what the difference is between the switch and other coherent combinations of channels~\cite{feix2015,abbott2020communication, guerin2016exponential,wood2021operational,gisin2005error,rubino2021quantum, rubino2021inferring,guerin2018shannon,chiribella2019quantum}. Finally, it would be intriguing to extend our upper bound to the case of an environment with memory, e.g.\ via the approach of Ref.~\cite{ciccarello2013collision}. 

\noindent{\bf {\em Acknowledgements--}} We acknowledge discussions with Qian Chen, Giulio Chiribella, Yutong Luo, Yi-Zheng Zhen and early discussions with David Felce, Sina Salek and Vlatko Vedral. We acknowledge support from the National Natural Science Foundation of China (Grants No. 12050410246, No. 1200509, No. 12050410245).

\bibliography{reference}

\appendix

\renewcommand{\thefigure}{A~\arabic{figure}}
\setcounter{figure}{0}
\renewcommand{\thetable}{A~\arabic{table}}
\setcounter{table}{0}
\renewcommand{\thepage}{A\arabic{page}} 
\setcounter{page}{1}
\renewcommand{\thesection}{A\arabic{section}}  
\setcounter{section}{0}

\section{Thermal state preserving condition and the second law} \label{app:thermal2ndlaw}
A thermal channel $\mathcal{E}$ is a completely positive, trace-preserving (CPTP) map that leaves the thermal state in question invariant.

Proposition: \textit{ The thermal state preserving condition, i.e., $\mathcal{E}(\tau)=\tau$,  is equivalent to the condition that the relative entropy between the state $\rho $ and the thermal state $\tau$ of the same system cannot increase, i.e., $ \Delta S(\rho || \tau):=S( \mathcal{E} ( \rho) || \tau ) - S(\rho || \tau) \leq 0 $ }.

\begin{proof}
$\Rightarrow$. According to the property of monotonicity of quantum relative entropy under CPTP maps, we have  
\begin{align}
 S(\rho || \tau) \geq S ( \mathcal{E} (\rho) ||  \mathcal{E} (\tau) ) =S ( \mathcal{E} (\rho) ||  \tau ),
 \end{align}
  where the thermal state preserving condition $ \mathcal{E} (\tau) =  \tau $ was used in the equality. This leads to 
  $$
  \Delta S(\rho || \tau) \leq 0.
  $$

$ \l\Leftarrow$. Consider initially $\rho=\tau$. Then, the change of the relative entropy is 
\begin{align}
\Delta S(\tau ||\tau)&= S(\mathcal{E} (\tau) ||\tau) -S(\tau ||\tau) \nonumber\\
&= S(\mathcal{E} (\tau) ||\tau) \leq 0.
\end{align}
Combining the above expression with the fact that relative entropy is non-negative gives us $S( \mathcal{E} (\tau) ||\tau) = 0.$ Finally, we arrive at $\mathcal{E} (\tau) =\tau $. 

\end{proof}

\section{Energy conservation gives an explicit parametrised two-qubit unitary collision model}\label{proof: energyconsv}

Consider a qubit collision model. Without loss of generality, we assume that $H_{E_i}= -\sigma_z $ and $H_s= -\vec{n} \cdot \vec{\sigma}$, and that there exists a one-qubit unitary transformation $ \tilde{U} $, such that
\begin{align}
  \tilde{U}: \,\, & \{ \ket{0},\ket{1}\} \rightarrow  \{ \ket{\phi}:=   \tilde{U}\ket{0} ,\ket{\phi^\perp }:=  \tilde{U} \ket{1}\}  ,\nonumber\\
   & \, \, \, \, \, \,\,  \,\,\, \sigma_z   \,\,\,\,\,\,\, \rightarrow \tilde{U} \sigma_z  \tilde{U}^\dag =\vec{n} \cdot \vec{\sigma} .
 \end{align}
That is to say that $ \ket{\phi}, \ket{\phi^\perp } $ are the eigenstates of $H_s$ with eigenvalues -1 and 1, respectively.

\begin{lem}\label{lem: energyconsv}
 For unitary collisions $U$ between qubits $M$ and $E_j$ obeying the energy conservation condition
 \begin{align} 
  [U_{ME_j}, H_M+H_{E_j}]=0,
  \end{align}
   it must be that $U_{ME_j}$ has the following parametrised form:
 \begin{align}\label{eq:collisionU}
 U_{ME_j}(\theta_1, \theta_2 ) = ( \tilde{U}\otimes \mathbb{I}) e^{i \theta_1 S }e^{i\theta_2\frac{ \sigma_z \otimes \sigma_z}{2}}   (\tilde{U}^\dag\otimes \mathbb{I}),
 \end{align}
 where $\tilde{U} $ is a  one-qubit unitary such that $\tilde{U} \sigma_z  \tilde{U}^\dag =\vec{n} \cdot \vec{\sigma}$ and 
the swap operator $S (\ket{\phi_1} \otimes \ket{\phi_2})= \ket{\phi_2} \otimes \ket{\phi_1}\, \forall \ket{\phi_1},\ket{\phi_2} \in \mathbb{C}^2 $.
   \end{lem}

\begin{proof}
First, we prove that, for qubit collisions, the energy conservation condition
 is equivalent to the global Gibbs state preserving condition, $ i.e.,$  
$$U  (\tau_s \otimes \gamma_{E_i} )U^\dag=\tau_s \otimes \gamma_{E_i}.$$
Denote $a:= e^{{\beta}}/( e^{{\beta}}+e^{{-\beta}})$.

$\Rightarrow $.  The proof is trivial for this direction. \\

$\Leftarrow$. 
The global Gibbs state preserving condition gives 
 \begin{align} \label{eq: gibbspre}
& U ( \gamma_A \otimes \tau_{E_i}) U^\dag \nonumber \\
= & U \Bigl( (a \ket{ \phi }\bra{\phi }+(1-a) \ket{\phi^\perp }\bra{ \phi^\perp } ) \otimes   \nonumber\\
&(a \ket{0}\bra{0}+(1-a) 
\ket{1}\bra{1} )\Bigl)  U^\dag \nonumber\\
=& \gamma_A \otimes \tau_{E_i} .
\end{align}
Demanding that Eq.~\eqref{eq: gibbspre} holds for all temperatures implies that the unitary interactions $U$ must preserve three subspaces: $ \ket{\phi, 0} \bra{\phi , 0},  \ket{\phi^\perp, 1} \bra{\phi^\perp ,1} $ and $ \ket{\phi ,1} \bra{\phi^\perp , 0}+ \ket{\phi^\perp , 0} \bra{\phi ,1}$. One sees that each subspace consists of the states with the same energy, which means that $U$ preserves energy, $i.e, [U, H_s+H_{E_i}]=0$ holds.

Next, we use the result from Ref~\cite{scarani2002thermalizing}, which gives an explicit parametrised form of  two-qubit unitaries obeying that global Gibbs state preserving condition.

Denote $S$ as the swap operator, $i.e.,$
\begin{align}
S=\left(                 
  \begin{array}{cccc}   
      1 & 0 &0  & 0 \\  
      0  &  0 & 1 & 0 \\
      0  &  1 & 0 & 0 \\
      0  &  0 & 0 & 1 \\  
      \end{array}
\right) .
\end{align}

 The general form of Gibbs state preserving unitaries~\cite{scarani2002thermalizing}, up to a global phase factor,  is $U(\theta_1, \theta_2):= U_{partialsw}(\theta_1) U_{phasefluct}(\theta_2)  $ where 
 \begin{align}
 & U_{partialsw}(\theta_1) = ( \tilde{U}\otimes \mathbb{I}) e^{i \theta_1 S }  (\tilde{U}^\dag\otimes \mathbb{I}) \nonumber\\
  =& c\mathbb{I} 
 +is \Bigl( \ket{\phi, 0}\bra{\phi, 0} + \ket{\phi^\perp, 1}\bra{ \phi^\perp , 1} 
+\ket{\phi, 1}\bra{\phi^\perp , 0} \nonumber\\
&+\ket{ \phi^\perp,  0}\bra{\phi ,1} \Bigl),\nonumber\\
&U_{phasefluct}(\theta_2) = (\tilde{U} \otimes \mathbb{I}) e^{i\theta_2\frac{ \sigma_z \otimes \sigma_z}{2}}  (\tilde{U}^\dag\otimes \mathbb{I}) \nonumber\\
=& e^{i\theta_2}  \ket{\phi, 0}\bra{\phi, 0} +e^{-i\theta_2}  \ket{\phi, 1}\bra{\phi, 1}  \nonumber\\
&+e^{-i\theta_2}  \ket{\phi ^\perp, 0}\bra{\phi^\perp, 0}+e^{i\theta_2}  \ket{\phi^\perp, 1}\bra{\phi^\perp, 1}.
  \end{align}
  It is straightforward to verify that $[  U_{partialsw}(\theta_1), U_{phasefluct}(\theta_2) ]=0$.
\end{proof}

 \section{Kraus decomposition of the quantum partial swap unitary}\label{krausdecomp}

Consider a unitary interaction $U_{ME}(\theta):=e^{i\theta S}=(c  \mathbb{I}+is S)$ where $S$ is the swap operator. If the heat bath is in the state $P_0=\ket{0}_E\bra{0}$, one has
\begin{align}
\tilde{E}_0=_E\bra{0}U\ket{0}_E=(c\id +is \ket{0}\bra{0}), \nonumber \\
\tilde{E}_1=_E\bra{1}U\ket{0}_E=i s \ket{0}\bra{1}  .\nonumber
\end{align}
If the heat bath in the state $P_1=\ket{1}_E\bra{1}$, one has
\begin{align}
\tilde{E}'_0=_E\bra{0}U\ket{1}_E= i  s \ket{1}\bra{0}) , \nonumber \\
\tilde{E}'_1=_E\bra{1}U\ket{1}_E=(c\id +is \ket{1}\bra{1}) . \nonumber
\end{align}

Hence, if the heat bath is in the state of  $\tau=q \ket{0}\bra{0}+(1-q)\ket{1}\bra{1}$, we have the set of Kraus operators
\begin{align}\label{eq:krausoperators}
K_1\equiv \sqrt{q} \tilde{E}_0=& \sqrt{q}   (c\id +is \ket{0}\bra{0}), \nonumber \\
K_2 \equiv \sqrt{q} \tilde{E}_1=& \sqrt{q} i  s \ket{0}\bra{1} ,\nonumber \\
K_3 \equiv \sqrt{1-q} \tilde{E}'_0=& \sqrt{1-q} i s \ket{1}\bra{0} ,\nonumber \\
K_4\equiv \sqrt{1-q} \tilde{E}'_1=& \sqrt{1-q} (c\id +is \ket{1}\bra{1}). 
\end{align}

\section{ Unitary dilation of control-system interaction is energy preserving} 
\label{sec: energyconsv}



Consider two thermal operations  $ \mathcal{E}_1, \mathcal{E}_2$ acting on the system $S$. $ \mathcal{E}_i, i=1,2$  are defined by $(U_i, \tau_{E_i})$ where $U_i$ is a unitary dilation of $\mathcal{E}_i$ satisfying
\begin{align}
[U_i, H_S+ H_{E_i}]  =0 ,
\end{align}
and 
$
 \tau_{E_i}:= e^{-\beta H_{E_i}}/Z_i 
 $
is the corresponding thermal environment state.  Further, the energy conservation condition ensures global Gibbs state conservation, i.e.,
\begin{align}
U_i ( e^{-\beta H_S} \otimes  e^{-\beta H_{E_i}} )U_i^\dag =  e^{-\beta H_S} \otimes  e^{-\beta H_{E_i}} .
\end{align}
Then, tracing out $E_i$, one arrives at $  \mathcal{E}_i(e^{-\beta H_S} ) =  e^{-\beta H_S}.$


Let  $C$ be the system controlling the orders of operations. Denote by $ \ket{\psi}, \ket{\psi^\perp} $ the two different eigenvectors of the Hamiltonian $H_C$ of the control system $C$. We then construct a unitary 
\begin{align}
 L= \ket{\psi }_C\bra{\psi } \otimes U_2  U_1 + \ket{\psi^\perp }_C\bra{\psi^\perp}\otimes U_1 U_2 . 
 \end{align}
  One can verify that, 
\begin{align}
&[L, H_C+H_S+H_{E_1}+H_{E_2}] \nonumber\\
=& \ket{\psi }_C\bra{\psi } \otimes [ U_2  U_1, H_S+H_{E_1}+H_{E_2}]  \nonumber\\
&+  \ket{ \psi^\perp }_C\bra{\psi^\perp} \otimes  [U_1  U_2, H_S+H_{E_1}+H_{E_2}] \nonumber\\
=& \ket{\psi }_C\bra{\psi } \otimes  \Bigl(U_2 [ U_1, H_S+H_{E_1}] + [ U_2, H_S  \nonumber\\
&+H_{E_2}] U_1 \Bigl) +  \ket{\psi^\perp }_C\bra{\psi^\perp} \otimes  \Bigl( U_1[  U_2, H_S \nonumber\\
&+H_{E_2}] + [U_1, H_S+H_{E_1} ]U_2 \Bigl) \nonumber\\
=&0,
\end{align}
where energy conservation conditions are used in last equality.
Thus, $L$ is an energy preserving unitary on the compound system $CSE_1E_2$.

We note for completeness that it follows that 
\begin{align}
&L( e^{-\beta H_C} \otimes  e^{-\beta H_{S}} \otimes  e^{-\beta H_{E_1}} \otimes e^{-\beta H_{E_2}} )L^\dag \nonumber\\
=& e^{-\beta H_C} \otimes  e^{-\beta H_{S}} \otimes  e^{-\beta H_{E_1}} \otimes e^{-\beta H_{E_2}} .
\end{align}


\section{Bound on information capacity from thermodynamics}
\label{infocapbound}

{\bf{Theorem~\ref{thm:main}}.} 
\textit{ For our model of two identical energy preserving thermal channels $ \{ \mathcal{C} \}$ undergoing a quantum switch, with $\lambda\in [0,1]$ representing the switch strength, $s$ the thermalisation strength, $q$ the probability of the ground state of the environment qubit thermal states, the final ({\em fin}) and initial ({\em in}) mutual informations between the untouched quantum system $A$ on the one hand and the control system $C$ plus the memory $M$ undergoing the channel on the other, obey
 \begin{align}
 I(A: CM^{ \fin }) \leq & \Bigl( c^4+ \lambda  (q^2-q+\frac{1}{2} )s^4 \Bigl)  I(A: M^\In) \nonumber\\
 &  + \lambda  G_{\geq 0},
\end{align}
where $ G_{\geq 0} := fI(A:D_0) +(1-f) I(A:D_1)$,$f=  1-  \frac{ (1-q^2)+(2q-1)(1-p)}{2} s^4 $ and the non-negative terms 
\begin{align}
I(A:D_0)=& S( p_{0} \rho_A+ (1-p_{0})  \rho^{(0)}_A )  
 - p_{0} S(\rho_A) \nonumber\\
 & -  (1-p_{0})  S(\rho^{(0)}_A ) , \nonumber\\
I(A:D_1)=& S( p_{1} \rho_A  + (1-p_{1})  \rho^{(1)}_A ) - p_{1} S(\rho_A) \nonumber\\
&-  (1-p_{1})  S(\rho^{(1)}_A ),
\end{align} 
 with $p_0 :=   (1+\frac{(q^2-2q) s^4}{2}) /f , p_1 = \frac{ (1-q^2)s^4}{2} /(1-f)$,  which involves virtual ancillary qubits $D_0$ and $D_1$ that are fully determined by the initial conditions. }

\begin{proof}

In our model, the heat bath's Hamiltonian $H_{E_j}=-\sigma_z\,\forall j$, and memory system $M$'s Hamiltonian $H_M=-\vec{n} \cdot \vec{\sigma}$, where $\vec{n}$ is a three-dimensional unit vector and $\sigma$ indicates Pauli-matrices. Without loss of generality, we assume the Hamiltonian of the two-dimensional control system $H_c=-\sigma_z$. The energy eigenstates of the control system are $\{ \ket{0}, \ket{1} \}$. 
 The initial state of the joint system  $CAM E_1E_2$ is 
 \begin{align}
&\tilde{U}_M^\dag \rho^\In \tilde{U}_M \nonumber\\
=&\sigma_C \otimes  \tilde{U}_M^\dag \rho^{\text{in}}_{AM} \tilde{U}_M \otimes  \tau_{E_1} \otimes \tau_{E_2}  \nonumber\\
 =& \Bigl(\lambda \ket{+}_c\bra{+}+(1-\lambda)\ket{0}_c\bra{0} \Bigl)  \otimes \Bigl( p \rho^{(0)}_A \otimes \ket{ 0 }_{M} \bra{0} 
\nonumber\\
&+(1-p) \rho^{(1)}_{A} \otimes \ket{1 }_{M}\bra{ 1} \Bigl) \otimes \tau_{E_1} \otimes \tau_{E_2},
\end{align} 
where $ \tilde{U} \sigma_z \tilde{U}^\dag  =\vec{n} \cdot \vec{\sigma} $ .

According to lemma~\ref{lem: energyconsv}, the energy conservation condition of thermal channels $\{ \mathcal{C}\}$ gives a parametrized form of the two-qubit unitary dilation:
\begin{align}
U_{ME_i}(\alpha,\beta) &=  ( \tilde{U}\otimes \mathbb{I}) e^{i \alpha S} e^{i\beta\frac{\sigma_z \otimes \sigma_z }{2}}  ( \tilde{U}^\dag \otimes \mathbb{I}) \nonumber\\
&:= ( \tilde{U}\otimes \mathbb{I}) U_i  ( \tilde{U}^\dag \otimes \mathbb{I}), \nonumber
\end{align}
where $S$ is the swap operator, i.e., $S \ket{\phi_1} \otimes \ket{\phi_2}= \ket{\phi_2} \otimes \ket{\phi_1}\, \forall \ket{\phi_1},\ket{\phi_2} \in \mathbb{C}^2 $.
  For the joint system $CAM E_1E_2$, the corresponding Kraus operator of $ \mathcal{S}( U_{ME_1}, U_{ME_2}) $ is \begin{align}
 L=& \ket{0}_c\bra{0} \otimes U_{ME_2}  U_{ME_1} + \ket{1}_c\bra{1}\otimes U_{ME_1} U_{ME_2} \nonumber\\
 =& \tilde{U} (\ket{0}_c\bra{0} \otimes U_2  U_{1} + \ket{1}_c\bra{1}\otimes U_{1} U_{2}) \tilde{U}^\dag \nonumber\\
 :=& \tilde{U} \tilde{L} \tilde{U}^\dag.
 \end{align}
 Then, the final state of $CAME_1E_2 $ is 
\begin{align}\label{eq: finalstate}
\rho^{\text{fin}}& = L \rho^{\text{\In}} L^\dag =  \tilde{U} \tilde{L} \tilde{U}^\dag \rho^\In \tilde{U}  \tilde{L} \tilde{U}^\dag.
\end{align}
Note that $ \tilde{L} \tilde{U}^\dag \rho^\In \tilde{U}  \tilde{L} $ is equivalent to the final state of $CAME_1E_2 $ when system $M$'s Hamiltonian $H_M= -\sigma_z$. Since we are interested in the mutual information between  $A$ and $CM$ and local unitaries do not change mutual information, we restrict ourselves to the case of $H_M=-\sigma_z$ in the following.
Moreover, $[e^{i\alpha S}, e^{i\beta\frac{\sigma_z \otimes \sigma_z }{2}}]=0$, $M$ and $E_i$ are initially diagonal in the computational basis. This implies that $e^{i\beta\frac{\sigma_z \otimes \sigma_z }{2}} $ does not impact the final states. Thus, we only need to consider $U_i=  e^{i \theta S_{ME_i}}, i=1,2,$ in the following.

 The Kraus operators of energy-preserving thermal channels $\mathcal{C}$ with a unitary dilation $U= e^{i\theta S}$ and thermal state $\tau$ are given in Appendix~\ref{krausdecomp}.
   The corresponding Kraus Operators for the resulting switched channels $  \mathcal{S}_{\sigma_C}( \mathcal{C}, \mathcal{C}) $ are 
  \begin{align}
   S_{ij}= & \ket{0}_C\bra{0} \otimes ( \mathbb{I}^A \otimes K_i^M)(\mathbb{I}^A \otimes K_j^M )  \nonumber\\
   & + \ket{1}_C\bra{1} \otimes  (\mathbb{I}^A \otimes K_j^M)(\mathbb{I}^A \otimes K_i^M ).
   \end{align}
 The final state of $CAM$ can be obtained by 
\begin{align}\label{eq: finalstateCAM}
\rho_{CMA}^\fin = \Tr_{E_1E_2} \rho^\fin = \mathcal{S}_{\sigma_C} (\mathcal{E}, \mathcal{E})(\rho_{CAM}^\In).
\end{align}
 
 \vspace{0.5cm}
{ \bf{Consider $\lambda=1$}}.  Based on Eq.~\eqref{eq: actionofswitch}, Eq.~\eqref{eq:krausoperators} and Eq.~\eqref{eq: finalstateCAM}, we further find that the final state of $CAM$ has the structure of 
 \begin{align}\label{eq: CMAstructure}
&\rho^{\text{ \fin }}_{CMA} \nonumber\\
=&  \ket{+0}_{CM}\bra{+0} \otimes \rho^{(1')}_A +\ket{+1}_{CM}\bra{+1} \otimes \rho^{(2')}_A \nonumber\\
&+  \ket{-0}_{CM}\bra{-0} \otimes \rho^{(3')}_A+ \ket{-1}_{CM}\bra{-1} \otimes \rho^{(4')}_A.
 \end{align}
 where 
 \begin{align} \label{eq: conditional1}
&  \rho^{(1')}_A  = \sum_{ij}\Tr_{CM} \ket{+0}_{CM}\bra{+0} \Bigl[ S_{ij}( \sigma_C \otimes \rho_{AM}^{\In}) S_{ij}^\dag   \Bigl]  \nonumber\\ 
 =&  \sum_{ij}\Tr_{CM} \ket{+0}_{CM}\bra{+0} S_{ij}  \Bigl(  \ket{+}_{C}\bra{+} \otimes ( p\ket{0}_{M}\bra{0} \otimes \rho^{(0)}_{A}  \nonumber\\
 & + (1-p)\ket{1}_{M}\bra{1} \otimes \rho^{(1)}_{A}) \Bigl) S_{ij}^\dag \nonumber\\
  =& p_0^{1'} \rho^{(0)}_{A} +  p_1^{1'} \rho^{(1)}_{A} .
 \end{align}
 Similarly, we have
 \begin{align}\label{eq: conditional2}
  \rho^{(2')}  =& p_0^{2'} \rho^{(0)}_{A} +  p_1^{2'} \rho^{(1)}_{A} , \nonumber\\
  \rho^{(3')}  =& p_0^{3'} \rho^{(0)}_{A} +  p_1^{3'} \rho^{(1)}_{A} , \nonumber\\
  \rho^{(4')}  =& p_0^{4'} \rho^{(0)}_{A} +  p_1^{4'} \rho^{(1)}_{A} .
 \end{align}
 The coefficients $ p_0^{1'},  p_1^{1'} $ are obtained by 
 \begin{align}\label{eq: cofficients1}
 p_0^{1'}= & \sum_{ij}\Tr_{CM} \ket{+0}_{CM}\bra{+0}  \nonumber\\
 & S_{ij} \Bigl(  \ket{+}_{C}\bra{+} \otimes ( p\ket{0}_{M}\bra{0}  \Bigl)  S^{\dag }_{ji}  \nonumber\\
   =& p(c^4+2c^2s^2 q +\frac{1}{2} s^4 q(1+q) ),  \nonumber  \\
 p_1^{1'}= & \sum_{ij}\Tr_{CM} \ket{+0}_{CM}\bra{+0} \nonumber\\
 &  S_{ij}  \Bigl(  \ket{+}_{C}\bra{+} \otimes (1-p)\ket{1}_{M}\bra{1}  \Bigl)  S^{\dag }_{ji}  \nonumber\\
   =& \frac{1-p}{2} s^2( 4c^2+s^2) q .
 \end{align}
 Similarly, one obtains the coefficients :
 \begin{align}\label{eq: coefficients2}
 p_0^{2'}  =& \frac{p}{2} s^2( 4c^2+s^2)(1-q) , \nonumber\\
  p_1^{2'}
   =& (1-p)\bigl(c^4+2c^2s^2 (1-q) +\frac{1}{2} s^4 (2-q)(1-q) \bigl) , \nonumber\\
    p_0^{3'}= & \frac{p}{2}  s^4q(1-q) , \ \  p_1^{3'}=  \frac{1-p}{2}s^4q  , \nonumber\\
 p_0^{4'}=&  \frac{p}{2}  s^4(1-q)  , \ \ p_1^{4'}=  \frac{1-p}{2}s^4q (1-q) .
 \end{align}

Substituting Eqs.\eqref{eq: conditional1}-\eqref{eq: coefficients2} in Eq.\eqref{eq: CMAstructure}, one can write down the explicit form of the final state of $CMA$
\begin{align}\label{eq: lambda1}
& \rho^{\lambda=1}_{CMA} \nonumber\\ 
=& \ket{-}_C\bra{-} \otimes s^4 \Bigl( \frac{ (1-q)^2}{2} (\sigma_x)_M \rho_{MA} (\sigma_x)_M + q(1-q)  \nonumber\\ 
 &\frac{ \mathbb{I}_M }{2}\otimes \rho_A + \frac{2q-1}{2} (1-p) \ket{0}_M \bra{0} \otimes  \rho^{(1)}_A  \Bigl) \nonumber\\ 
&+ \ket{+}_C\bra{+} \otimes \Bigl( (c^4+ \frac{(1-q)^2}{2}s^4) \rho^{\text{in }}_{MA} +(2c^2s^2+ \frac{s^4}{2}) \nonumber\\
& \tau_M \otimes \rho_A 
+ s^4 \frac{2q-1}{2} p \ket{0}_M \bra{0} \otimes  \rho^{(0)}_A \Bigl) .
\end{align} 

\vspace{0.5cm}
{\bf{Consider $\lambda \in [0,1]$}}. It is straightforward to write down the final state of $CMA$
\begin{align}\label{eq: finalCMA}
  \rho^{\fin}_{CMA}= & \lambda  \rho^{\lambda=1}_{CMA} +(1-\lambda) \ket{0}_c\bra{0} \otimes (c^4\rho_{MA}^\In \nonumber\\
  & +(1-c^4) \tau_M\otimes \rho_A ).
\end{align}
Adding an ancilla system $B$, we have 
\begin{align}\label{eq: finalCMAB}
  &\rho^{\fin}_{BCMA} \nonumber\\
  = & \lambda \ket{0}_B\bra{0} \otimes  \ket{-}_C\bra{-} \otimes s^4 \Bigl( \frac{ (1-q)^2}{2} (\sigma_x)_M \rho_{MA} (\sigma_x)_M  \nonumber\\ 
 &+ q(1-q) \frac{ \mathbb{I}_M }{2}\otimes \rho_A + \frac{2q-1}{2} (1-p) \ket{0}_M \bra{0} \otimes  \rho^{(1)}_A  \Bigl) 
 \nonumber\\ 
&+ \lambda \ket{1}_B\bra{1} \otimes  \ket{+}_C\bra{+} \otimes \Bigl( (c^4+ \frac{(1-q)^2}{2}s^4) \rho^{\text{in }}_{MA}  \nonumber\\
&+(2c^2s^2+ \frac{s^4}{2}) \tau_M \otimes \rho_A 
+ s^4 \frac{2q-1}{2} p \ket{0}_M \bra{0} \otimes  \rho^{(0)}_A \Bigl) \nonumber\\
& +(1- \lambda) \ket{2}_B\bra{2} \otimes  \ket{0}_C\bra{0} \otimes  (c^4\rho_{MA}^\In 
   +(1-c^4) \tau_M\otimes \rho_A ).
\end{align}

To bound the final mutual information between $A$ and $CM^\text{fin}$, we apply the quantum data-processing inequality and the chain rule 
\begin{align}\label{eq: chainrule}
&  I(A: CM^{ \text{fin} }) \leq   I(A: BCM^{ \text{fin} }) \nonumber\\
  &=I(A:B)+ I(A: CM^{\text{fin} } |B ) .
 \end{align} 

First, we bound the term $ I(A:B ) $.  Tracing out $CM$ from Eq.\eqref{eq: finalCMAB}, we obtain the state of $BA$ as follows
\begin{align}
  &\rho^{\fin}_{BA} \nonumber\\
=& \lambda \Bigl(s^4 \frac{ 1-q^2}{2} \ket{0}_B\bra{0} +  (c^4+ \frac{1+(1-q)^2}{2}s^4+2c^2s^2 ) \nonumber\\
&  \ket{1}_B \bra{1} \Bigl) \otimes \rho_A + \lambda \frac{2q-1}{2} s^4  \Bigl( (1-p)\ket{0}_B \bra{0}  \otimes  \rho^{(1)}_A  \nonumber\\
&+ p \ket{1}_B\bra{1} \otimes  \rho^{(0)}_A \Bigl)+ (1-\lambda)  \ket{2}_B\bra{2} \otimes  \rho_A.
\end{align}
Adding an ancilla system $D$ to $ \rho^{\text{fin}}_{BA} $, we have 
\begin{align}\label{eq: DCA}
& \rho^{\text{fin}}_{DBA} \nonumber\\
=&  \ket{0}_D\bra{0} \otimes \Bigl( \lambda s^4 \frac{ 1-q^2}{2} \ket{0}_B\bra{0} + \lambda (c^4+ \frac{1+(1-q)^2}{2}s^4 \nonumber\\
&+2c^2s^2 )  \ket{1}_B\bra{1} +(1-\lambda)  \ket{2}_B\bra{2} \Bigl) \otimes \rho_A  + \lambda \ket{1}_D\bra{1} \nonumber\\
&\otimes  (q-\frac{1}{2} ) s^4  \Bigl( (1-p)\ket{0}_B\bra{0}  \otimes  \rho^{(1)}_A + p \ket{1}_B\bra{1} \otimes  \rho^{(0)}_A \Bigl)  .
\end{align}
It is straightforward to prove that $ I(A:D)_{  \rho^{\fin}_{DBA}}=0$. Then, one can calculate
\begin{align}\label{eq: mutualAC}
& I(A: B)_{ \rho^{\fin}_{BA}} \leq  I(A: DB)_{\rho^{\fin}_{DBA}} \nonumber\\
= & I(A:D)_{  \rho^{\text{fin}}_{DBA}} + I(A: B|D)_{\rho^{\text{fin}}_{DBA}} \nonumber\\
= &  \lambda s^4 (q- \frac{1}{2}) I(A : M^\In ),
\end{align}
where the quantum data-processing inequality is used in the first inequality.

Next, we proceed to bound the second term $ I(A: CM^{\text{fin} } |B ) $. According to the definition of conditional mutual information, one can write down 
\begin{align}
&I(A: CM^{\text{fin} } | B) =   P_0 I(A:CM^{\text{fin} } | \ket{0}_B)  \nonumber\\
& +P_1 I(A: CM^{\text{fin} } | \ket{1}_B) + P_2 I(A: CM^{\text{fin} } | \ket{2}_B) ,
\end{align}
where 
\begin{align}
  P_0 & = \lambda (1-  \frac{ (1-q^2)+(2q-1)(1-p)}{2} s^4 ) , \nonumber\\
   P_1 &=  \lambda -  P_0 ,P_2 = 1- \lambda . 
 \end{align}
 
We then introduce an ancilla system $D_{0}$ for $ \rho_{ ACM| \ket{0}_B}$ such that 
 \begin{align}
& \rho_{D_0 ACM| \ket{0}_B} \nonumber\\
=& \frac{\lambda}{   P( \ket{+}_C )} \ket{+}_c\bra{+} \otimes \Bigl( (c^4+\frac{(1-q)^2}{2}s^4) \ket{0}_{D_0}\bra{0} \otimes \rho^{\text{in }}_{MA} \nonumber\\ 
&+(2c^2s^2+ \frac{s^4}{2})\ket{1}_{D_0}\bra{1} \otimes  \tau_M \otimes \rho_A  + \frac{(2q-1)p s^4}{2}  \nonumber\\ 
&  \ket{2}_{D_0} \bra{2}\otimes \ket{0}_M \bra{0} \otimes  \rho^{(0)}_A \Bigl).
\end{align}
Tracing out $CM$, one obtains
 \begin{align}
& \rho_{D_0  A| \ket{0}_B} \nonumber\\
=& \frac{\lambda }{   P_0}\Bigl( (c^4+\frac{(1-q)^2}{2}s^4) \ket{0}_{D_0}\bra{0} \otimes \rho_{A} +(2c^2s^2+ \frac{s^4}{2})\nonumber\\
&\ket{1}_{D_0}\bra{1}  \otimes \rho_A + \frac{(2q-1)p s^4}{2}  \ket{2}_{D_0} \bra{2} \otimes  \rho^{(0)}_A \Bigl) .  \nonumber
\end{align}
The mutual information between $A$ and $D_0$ is 
\begin{align}
I(A: D_0 )=& S(p_0 \rho_A+(1-p_0) \rho^{(0)}_A ) \nonumber\\
&- p_0S( \rho^{}_A)-(1-p_0) S(\rho^{(0)}_A ) ,
\end{align} 
where $p_0 :=\lambda (1+\frac{(q^2-2q) s^4}{2}) /P_0 )$. 
Then, 
\begin{align} \label{eq: mutualAM+}
& I(A:CM)_{\rho_{ ACM| \ket{0}_B}} = I(A:M)_{\rho_{ ACM| \ket{0}_B}}  \nonumber\\
 \leq & I(A: D_0M)_{\rho_{ D_0 ACM| \ket{0}_B}} 
= I(A: D_0 )+I(A: M|D_0) \nonumber\\
=& \frac{\lambda}{P_0} ( c^4+ \frac{ (1-q)^2}{2  } s^4)  I(A: M^\text{in}) + I(A: D_0 ) ,
\end{align}
where the quantum data-processing inequality is used in the first inequality. 

Again, we introduce an ancilla system $D_{1}$ for $ \rho_{C AM| \ket{1}_B}$ such that 
\begin{align}
  & \rho_{D_1 CMA | \ket{1}_B} \nonumber\\ 
 = & \frac{ \lambda s^4}{   P_1}\Bigl( \frac{ (1-q)^2}{2}  \ket{0}_{D_1}\bra{0} \otimes  (\sigma_x)_M \rho_{MA} (\sigma_x)_M  + q(1-q)  \nonumber\\
 &\ket{1}_{D_1}\bra{1} \otimes \frac{ \mathbb{I}_M }{2}\otimes \rho_A  + \frac{2q-1}{2} (1-p)  \ket{2}_{D_1}\bra{2} \otimes \nonumber\\
 &\ket{0}_M \bra{0} \otimes  \rho^{(1)}_A \Bigl) \otimes \ket{-}_c\bra{-}. 
\end{align}
Tracing out $CM$, one obtains 
\begin{align}
  & \rho_{D_1 A | \ket{1}_B} \nonumber\\ 
 = & \frac{ \lambda s^4}{   P_1}\Bigl( \frac{ (1-q)^2}{2}  \ket{0}_{D_1}\bra{0} \otimes   \rho_{A} + q(1-q)  \nonumber\\
 & \ket{1}_{D_1}\bra{1} \otimes \rho_A + \frac{2q-1}{2} (1-p)  \ket{2}_{D_1}\bra{2} \otimes  \rho^{(1)}_A \Bigl).  \nonumber
\end{align} 
The mutual information between $A$ and $D_1$ is 
\begin{align}
I(A: D_1)=& S(p_1 \rho_A+(1-p_1) \rho^{(1)}_A ) \nonumber\\
&- p_1 S( \rho_A)-(1-p_1) S(\rho^{(1)}_A ) ,
\end{align} 
where $p_1 := \frac{\lambda (1-q^2)s^4}{2} /P_1$. Similarly, we have
\begin{align}\label{eq: mutualAM-}
& I(A: CM)_{\rho_{ CMA| \ket{1}_B}} = I(A: M)_{\rho_{ CMA| \ket{1}_B}}   \nonumber\\
 \leq & I(A: D_1C M)_{\rho_{ D_1 CAM| \ket{1}_B }} = I(A: D_1 )+I(A: M|D_1) \nonumber\\
=&  \frac{ \lambda  (1-q)^2s^4}{2  P_1  }  I(A: M^\text{in}) + I(A: D_1).
\end{align}

Similarly, we introduce an ancilla system $D_{2}$ for $ \rho_{ CAM| \ket{2}_B}$ such that 
\begin{align}
   \rho_{ CD_2MA | \ket{2}_B}= &  \ket{0}_C\bra{0} \otimes  \Bigl(c^4 \ket{0}_{D_2}\bra{0} \otimes \rho_{MA}^\In  \nonumber\\
  & +(1-c^4)  \ket{1}_{D_2}\bra{1} \otimes\tau_M\otimes \rho_A \Bigl). 
\end{align}
Then, the mutual information $ I(A: CM^{\text{fin} } | \ket{2}_B)$ is upper bounded by
\begin{align}\label{eq: mutualAB2}
& I(A: CM^\fin )_{\rho_{ CMA| \ket{2}_B}}    
 \leq  I(A: CD_2 M)_{\rho_{ CD_2 AM| \ket{1}_B }} \nonumber\\
 =& I(A: D_2 )+I(A: M|D_2) =c^4  I(A: M^\text{in}) .
\end{align}

Finally, combining Eq.\eqref{eq: mutualAC}, Eq.\eqref{eq: mutualAM+}, Eq.\eqref{eq: mutualAM-}, Eq.\eqref{eq: mutualAB2} with Eq.~\eqref{eq: chainrule},  we obtain
\begin{align}
 & I(A: CM^{ \text{fin} }) \leq   I(A: BCM^{ \text{fin} }) \nonumber\\
  =&I(A:B)+ I(A: CM^{\text{fin} } |B ) 
 \nonumber\\
  \leq & \lambda (q-\frac{1}{2})s^4I(A:M^{\text{in}}) +   P_0 I(A:M^{\text{fin} } | \ket{0 }_B) \nonumber\\
 &+P_1 I(A:M^{\text{fin} } | \ket{1}_B)+P_2 I(A:M^{\text{fin} } | \ket{2}_B) \nonumber\\
 \leq & ( c^4+ \lambda (q^2-q+\frac{1}{2})s^4)  I(A: M^\text{in}) +   P_0 I(A:D_0) \nonumber\\
  & +   P_1 I(A:D_1)  ,
\end{align}
where $ P_0 : = \lambda( 1-  \frac{ (1-q^2)+(2q-1)(1-p)}{2} s^4 ), P_1:= \lambda-P_0  $ are the probabilities conditioning on $B$ , and 
\begin{align}\label{eq: virtualmutual}
 &I(A: D_{0/1} ) :=  S( p_{0/1} \rho_A+ (1-p_{0/1})  \rho^{(0/1)}_A )  \nonumber\\
 &- p_{0/1} S(\rho_A)-  (1-p_{0/1})  S(\rho^{(0/1)}_A )
\end{align} with  $ p_0:=  \lambda (1+\frac{(q^2-2q) s^4}{2}) /P_0,  p_1 = \lambda \frac{(1-q^2)s^4}{2} /P_1$. 

\end{proof}

{\bf{Infinite temperature case } } In this case, the thermal state is $\tau= \frac{\mathbb{I}}{2}$,  $i.e., q=1/2$. Substituting $q=1/2$ into the bound, one has
\begin{align}
 & P_0 = \lambda (1-  \frac{ 3 }{8} s^4),  P_1 =   \frac{ 3 }{8} \lambda s^4,  p_0 = p_1= 1 .
\end{align}
Thus, according to Theorem~\ref{thm:main}, the mutual information $I(A: CM^\fin)$ is bounded by
\begin{align}
  I(A: CM^{ \text{fin} }) \leq  ( c^4+  \frac{ \lambda }{4}s^4)  I(A: M^\text{in}) .
 \end{align}

{\bf{Zero temperature case:} } In this case, the thermal state is $\tau= \ket{0}\bra{0}$,  $i.e., q=1$. According to Theorem~\ref{thm:main}, the mutual information $I(A: CM^\fin)$ is bounded by
\begin{align}
 & I(A: CM^{ \text{fin} })  \nonumber\\
 \leq & ( c^4+  \frac{ \lambda }{2}s^4)  I(A: M^\text{in}) +   P_0 I(A:D_0),
\end{align}
where $ P_0  = \lambda  (1-  \frac{ 1-p}{2} s^4 ) $ and 
\begin{align}
 I(A: D_{0} ) := & S( p_0 \rho_A+ (1-p_0 )  \rho^{(0)}_A )  \nonumber\\
 &- p_0 S(\rho_A)-  (1-p_0 )  S(\rho^{(0)}_A ) \nonumber
\end{align} with  $ p_0 = \lambda  (1-\frac{ s^4}{2}) /P_0$.

\section{Violation of the first law of thermodynamics}
\label{app:decoherent}

  \begin{figure}    
   \centering
  \includegraphics[width=0.48\textwidth]{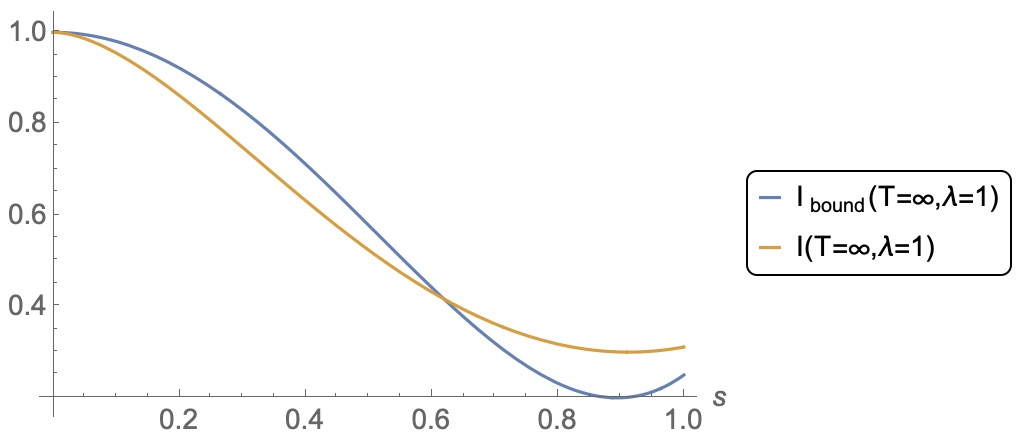} 
    \caption{Demonstration of violation of the bound by the use of an energy-altering channel. The blue line is the bound and the orange line is the mutual information achieved by two switched channels, one being associated with a partial swap unitary and the other a `partial CNOT'. }
        \label{fig:violation} 
\end{figure}

We take the Hamiltonians of the system and the environment to be $H_M=H_E= \sigma_z$.  When the temperature is infinity, the thermal state is then $\tau:= \frac{\mathbb{I}}{2}$. Consider the CNOT gate 
$$ 
V:=  \ket{0}\bra{0} \otimes \mathbb{I} + \ket{1}\bra{1} \otimes \sigma_x,
$$
Then, $e^{i \theta V} = \cos \theta \mathbb{I} +i  \sin \theta V := c \mathbb{I} +i  s V $. One can check that 
\begin{align}
&e^{i \theta V}\left(  \frac{\mathbb{I}_M }{2} \otimes  \frac{\mathbb{I}_E }{2} \right) e^{-i \theta V} =  \frac{\mathbb{I}_M }{2} \otimes   \frac{\mathbb{I}_E}{2}.
\end{align}
Tracing out $E$ in the above equation results in $ \mathcal{N} (  \frac{\mathbb{I}}{2} )= \frac{\mathbb{I}}{2}. $  Thus, $\mathcal{N}$ is a thermal channel and the second law is obeyed.  The set of Kraus operators for $\mathcal{N}$ is
 \begin{align}
  \{ (c \mathbb{I}+ i s \ket{0}\bra{0}),   i s\ket{1}\bra{1} \}. 
  \end{align}
However, one sees that 
$$
[e^{i \theta V}, H_E+H_M] \neq 0,
$$
i.e., the first law of thermodynamics is violated. Recall that the Kraus operators for the quantum partial swap unitary given in the Appendix~\ref{krausdecomp}. We see that some elements of the two sets of Kraus operators do not commute, enabling  information capacity increase. 

As shown in Figure~\ref{fig:violation}, we plotted the mutual information $I(A:CM^\fin)$  achieved by two cases of switched channels, one being associated with a partial swap unitary and the other a `partial CNOT'. The figure shows that our bound is violated in some region by the `partial CNOT' case, consistent with the energy conservation violation of that channel.

\section{Proof of Theorem~\ref{superchan}} \label{proofsuperchan}

Recall that the quantum switch $\mathcal{S}_{\sigma_C}$ is a supermap~\cite{chiribella2008quantum}, which takes two quantum channels $ \mathcal{C}_1,\mathcal{C}_2 $ as input and output another quantum channel $ \mathcal{S}_{\sigma_C} ( \mathcal{C}_1, \mathcal{C}_2 )$.  In the language of Kraus operators~\cite{nielsen2002quantum}, if we denote the Kraus operators of the channel $ \mathcal{C}_1 $ as $\{ K^{(1)}_i\}$ and $ \mathcal{C}_2 $ as $ \{ K^{(2)}_j \} $, the action of the channel $ \mathcal{S}_{\sigma_C} ( \mathcal{C}_1, \mathcal{C}_2 )$ is expressed by
\begin{align}\label{eq: actionofswitch}
 \mathcal{S}_{\sigma_C} ( \mathcal{C}_1, \mathcal{C}_2 ) (  \rho  )= \sum_{ij} S_{ij} (\sigma_C \otimes \rho  ) S_{ij}^\dag
\end{align}
where $\sigma_C$ is the state of the control system, and $S_{ij}$ denotes the Kraus operators for $ \mathcal{S}_{\sigma_C} ( \mathcal{C}_1, \mathcal{C}_2 )$, namely,
\begin{align}
S_{ij}=|0\rangle_C \langle 0| \otimes K^{(2)}_j K^{(1)}_i +  |1\rangle_C \langle 1| \otimes K^{(1)}_i K^{(2)}_j \ .
\end{align}

Theorem \ref{superchan}:  \textit{
If $\mathcal{C}_1(\tau_M)=\mathcal{C}_2(\tau_M)=\tau_M$, where $\tau_M$ is the thermal state on $M$ then, for any state $\sigma_C$, $\Tr_C\mathcal{S}_{\sigma_C} ( \mathcal{C}_1, \mathcal{C}_2 ) ( \tau_M )=\tau_M$. }
\begin{proof}
Denote by $\sigma_C^{ij}$ the matrix element of $i$th row and $j$th column of $\sigma_C$. The action of the channel $\Tr_C\mathcal{S}_{\sigma_C}$ on the thermal state $\tau_M$ is given by 
\begin{align}
&\Tr_C\mathcal{S}_{\sigma_C}(\mathcal{C}_1, \mathcal{ C}_2)(\tau_M) \nonumber\\
=& \Tr_C\sum_{i,j} S_{ij}(\sigma_C \otimes \tau_M) S_{ij}^\dag \nonumber \\
=&  \sigma_C^{00}\sum_{i,j} K^{(2)}_j K^{(1)}_i \tau_M K^{(1) \dag}_i K^{(2) \dag}_j \nonumber \\
&+  \sigma_C^{11} \sum_{i,j} K^{(1)}_i K^{(2)}_j \tau_M K^{(2) \dag}_j K^{(1) \dag}_i \nonumber\\
=& \sigma_C^{00}\mathcal{C}_2\circ \mathcal{C}_1 (\tau_M)+\sigma_C^{11}\mathcal{C}_1\circ \mathcal{C}_2 (\tau_M) \nonumber\\
=&\tau_M,
\end{align}
where $\mathcal{C}_1(\tau_M)=\mathcal{C}_2(\tau_M)=\tau_M$ is used in the last equality.
\end{proof}
  
\section{Free energy of coherence of the control} \label{app:Fcoh}

Denote by $H$ the Hamiltonian of the system. The free energy of coherence of the system is defined as~\cite{lostaglio2015description}
\begin{align}
F_{coh}(\sigma, H) :=kT \Bigl(S( \mathcal{D}_H (\sigma)) - S(\sigma) \Bigl),
\end{align}
where $  \mathcal{D}_H $ is the map that kills all the off-diagonal elements in the energy eigenbasis. The relation between the free energy $F(\sigma, H )$ and the free energy of coherence is 
\begin{align}
F(\sigma, H)&:= \Tr (\sigma H)-kTS(\sigma) \nonumber\\
&=F(  \mathcal{D}_H (\sigma),H)+ F_{coh}(\sigma, H).
\end{align}

Let the Hamiltonian of the system $H=\sigma_z$, and the initial state of the system $\sigma = \lambda \ket{+}\bra{+} + (1-\lambda ) \ket{0}\bra{0}, \lambda \in [0,1]$. Then, $ \mathcal{D}_H(\sigma) = \lambda \frac{ \mathbb{I}}{2} +(1- \lambda) \ket{0}\bra{0}$. The set of eigenvalues of $  \mathcal{D}_H(\sigma) $ is $\{ 1- \frac{\lambda}{2},  \frac{\lambda}{2} \}$. 
The set of eigenvalues of $  \sigma $ is $\{ \frac{1}{2}(1\pm \sqrt{1-2\lambda+\lambda^2})  \}$. 
Then, applying a Taylor expansion and some manipulations, we arrive at 
\begin{align}
F_{coh}(\sigma, H) \geq kT \frac{\lambda^2}{\ln16} .
\end{align}

\section{Interpolating between switch being ON and OFF}
~\label{sec:purecontrol}
The main text utilized a control parameter $\lambda \in [0,1]$ to denote the level of mixture of the state $\sigma_c= \lambda \ket{+}\bra{+} + (1-\lambda) \ket{0}\bra{0}$ of the control. Naturally, when $\lambda$ is closer to one, the control system has more coherence -- regulating the processing orders of the two energy preserving thermal channels. In this Section, we explore whether a second method for gradually turning a switch ON and OFF, used in Ref.~\cite{ebler2018enhanced}, yields the same conclusions. In this scenario, the control state is pure and reads $\ket{ \Psi}_c = \sqrt{ \alpha} \ket{0} + \sqrt{1- \alpha }\ket{1}, \alpha \in [0,1]  $. Here, the switch is maximally turned on for $\alpha=1/2$, such that $\ket{ \Psi}_c =\ket{+}$.

 Following the same procedure  as in the Appendix~\ref{infocapbound} but replacing the control state with $ \ket{\Psi }_c$,  we calculate the final state of the control system $C$, the record of the message $A$ and the system undergoing the channel $M$ as 
  \begin{align}\label{eq: CAM2}
\rho^{\text{ \fin }}_{CMA} =& \mathcal{S}_{\ket{\Psi}_c}( \mathcal{ C}, \mathcal{ C})(\rho_{AM}) = \sum_{i,j} S_{ij} \rho_{AM} S_{ij}^\dag \nonumber \\
=&  \alpha |0\rangle_c \langle 0| \otimes \sum_{i,j} K^{M}_j K^{M}_i \rho_{AM} K^{M \dag}_i K^{M \dag}_j \nonumber \\
+ & (1-\alpha) |1\rangle_c \langle 1| \otimes \sum_{i,j} K^{M}_i K^{M}_j \rho_{AM} K^{M \dag}_j K^{M \dag}_i \nonumber \\
+ & \sqrt{ \alpha(1- \alpha)} |0\rangle_c \langle 1| \otimes \sum_{i,j} K^{M}_j K^{M}_i \rho_{AM} K^{M \dag}_j K^{M \dag}_i \nonumber \\
+ & \sqrt{ \alpha(1-\alpha) } |1\rangle_c \langle 0| \otimes \sum_{i,j} K^{M }_i K^{M}_j \rho_{AM} K^{M \dag}_i K^{M \dag}_j .
\end{align} 
 Note that here we have used the shorthand $K^M_i:=\mathbb{I}_A \otimes (K_i)_M$. We then calculate the two summations appearing in the above expression:
  \begin{align}
&\sum_{i,j} K^{M}_j K^{M}_i \rho_{AM} K^{M \dag}_i K^{M \dag}_j = c^4 \rho^\In_{AM}+(1-c^4) \rho_{A} \otimes \tau_M ,\nonumber\\
 & \sum_{i,j} K^{M}_i K^{M}_j \rho_{AM} K^{M \dag}_i K^{M \dag}_j = (c^2 +s^2 q)^2 \rho^{(0)}_A \otimes p \ket{0}_M\bra{0}  \nonumber\\
 &+ (c^2 +s^2(1- q))^2 \rho^{(1)}_A \otimes (1-p) \ket{1}_M\bra{1} + 2c^2s^2(1-q) \nonumber\\
 & \rho^{(0)}_A \otimes p \ket{1}_M\bra{1} +2c^2s^2 q \rho^{(1)}_A \otimes (1-p) \ket{0}_M\bra{0}  \nonumber\\
&= (c^2 +s^2(1- q))^2 \rho^\In_{AM}  + 2c^2s^2(1-q)  (\sigma_x)_M \rho^\In_{AM} (\sigma_x)_M \nonumber\\
& + (ps^4+2s^2c^2)(2q -1) \rho^{(0)}_A \otimes  \ket{0}_M\bra{0}.
\end{align} 
 Substituting the two summations into Eq. (A91), we obtain
\begin{align}\label{eq:CMAsta2}
&\rho^{\text{ \fin }}_{CMA} \nonumber\\
=&  (\alpha |0\rangle_c \langle 0| +  (1-\alpha) |1\rangle_c \langle 1|)\otimes \Bigl(c^4 \rho^\In_{AM}   +(1-c^4) \rho_{A} \otimes \tau_M \Bigl) \nonumber\\
& +  \sqrt{ \alpha(1- \alpha)}( |0\rangle_c \langle 1|+|1\rangle_c \langle 0| ) \otimes  \Bigl( (c^2 +s^2(1- q))^2 \rho^\In_{AM}  \nonumber\\
& + 2c^2s^2(1-q)  (\sigma_x)_M \rho^\In_{AM} (\sigma_x)_M  + (ps^4+2s^2c^2)(2q -1) \nonumber\\
& \rho^{(0)}_A \otimes  \ket{0}_M\bra{0} \Bigl) .
\end{align}

 We can break the general control system state into three cases:

\begin{enumerate}

\item[i)]   {\em Switch is fully ON: $\alpha=\frac{1}{2}$, i.e,} $\ket{\Psi}_c =\ket{+} $. By inspection Eq.(A93) and Eq.(A514) denote the same state, that is, this corresponds to $\lambda=1, \sigma_c =\ket{+}\bra{+}$. Thus our bound~\eqref{thm:main}  is directly applicable to this case.   

\item[ii)]  {\em Switch is fully OFF: $\alpha=0/1$, i.e,} $\ket{\Psi}_c =\ket{0}/\ket{1} $. Here, there is no superposition of orders. This corresponds to the switch being off, i.e., $\lambda=0, \sigma_c =\ket{0}\bra{0}$.  By inspection Eq.~(A93) and Eq.~(A514) denote the same state in this case. Thus our bound is directly applicable to this case as well.   

\item[iii)]   {\em Other values of $\alpha $.} The final state of the compound system $CMA$ expressed by Eq.~(A93) is not identical to that of Eq.~(A514). Nevertheless, the mutual information after the interaction is very similar whether the control system is taken as pure or mixed, as shown in Figure A2. 
\end{enumerate}
We conclude that the two possible choices of the general state of the control system lead to qualitatively and quantitatively very similar behaviour and focus on the case of the mixed control as it gives neater mathematical expressions.

  \begin{figure}    
   \centering
  \includegraphics[width=0.48\textwidth]{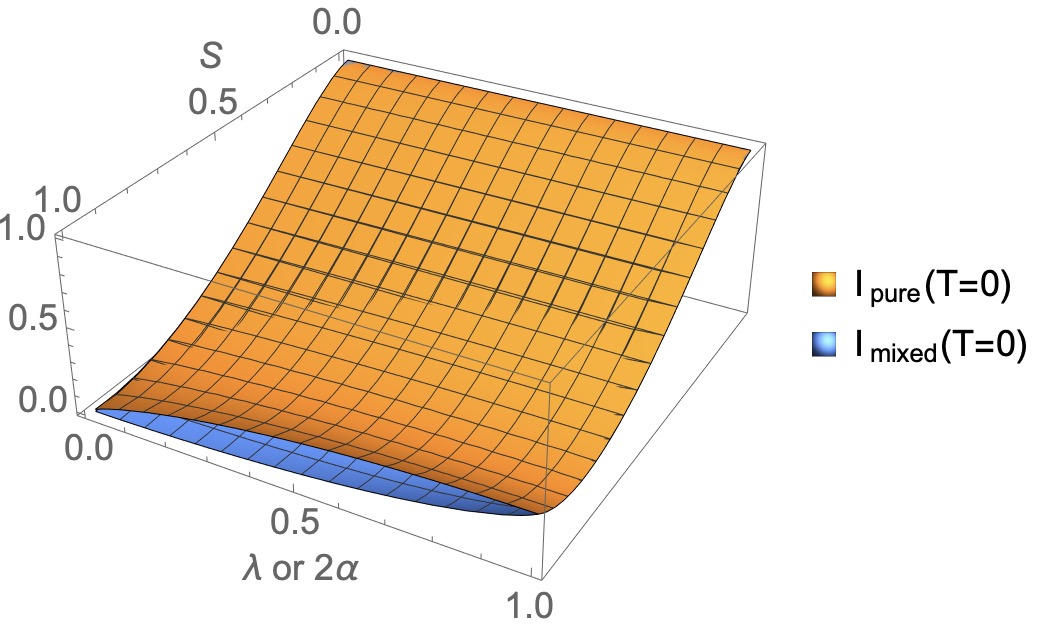} 
    \caption{The mutual information between memory and post-interaction system plus control is very similar whether the control is in the mixed state $\sigma_c= \lambda \ket{+}\bra{+} + (1-\lambda) \ket{0}\bra{0}, \lambda \in [0,1]$ vs. the pure state $\ket{ \Psi}_c = \sqrt{ \alpha} \ket{0} + \sqrt{1- \alpha }\ket{1}, \alpha \in [0,1/2]  $. The plot is for the case of the temperature $T=0$ (we observe similar behaviour for all temperatures). $s$ is a thermalisation strength parameter.}
        \label{fig:control} 
\end{figure}

\end{document}